\documentclass[conference, final]{IEEEtran}
\usepackage[pdftex]{graphicx}
\usepackage{pgf}
\usepackage{tikz}
\usepackage{epstopdf}
\usepackage{xcolor}
\graphicspath{{./}}
\DeclareGraphicsExtensions{.eps}
\usepackage[cmex10]{amsmath}
\usepackage{array}
\usepackage{mdwmath}
\usepackage{mdwtab}
\usepackage{subfig}
\usepackage{psfrag}
\usepackage[]{cite}
\usepackage{amssymb}
\usepackage{tikz}
\usetikzlibrary{calc,shapes,arrows}
\usepackage{verbatim}
\usepackage{graphicx}					
\usepackage{listings} 					
\usepackage{amsmath}					
\usepackage{amssymb}					
\usepackage{mathrsfs}					
\usepackage{amsthm}

\newcommand{\gd}{\mathcal{N}}
\newcommand{\E}{\mathbf{E}}

\theoremstyle{plain}

\newtheorem{theorem}{Theorem}

\theoremstyle{plain}

\begin{document}

\title{Communicating One Bit over a Delay Constrained Gaussian MIMO Channel with Feedback}

\author{ 
	\IEEEauthorblockN{Bo Bernhardsson}
	\IEEEauthorblockA{Department
	of Automatic Control\\ Lund University\\ Box 118, 221 00 Lund, Sweden\\ e-mail: bob@control.lth.se}\\
	
	\IEEEauthorblockN{Ather Gattami}
    	\IEEEauthorblockA {Bitynamics Research\\ 112 19 Stockholm, Sweden\\ atherg@gmail.com} 
}


\maketitle

\begin{abstract}

The energy-optimal scheme is found for communicating one bit over a memoryless Gaussian channel with an ideal feedback channel. It is assumed that the  channel is allowed to be used at most $N$ times before decoding. The optimal coding/decoding strategy is derived by dynamic programming. It is found that feedback gives a significant performance gain and that the optimal strategies are discontinuous. It is also shown that most of the performance increase can be obtained even with a one-bit feedback channel. The optimal scheme is compared with the strategy by Kailath-Schalkwijk and is found to be significantly more effective. 
For the case of a diagonal MIMO channel where measurement noise variances are equal along the sub channels we also show that the problem can be reduced to the previous case of transmitting one bit over a scalar feedback channel.
\footnote{The research was supported by the ELLIIT Excellence Center and by
the Swedish Research  Council through the LCCC Linneaus Center.}
\end{abstract}
\IEEEpeerreviewmaketitle

\IEEEpeerreviewmaketitle

\section{Introduction}
\IEEEPARstart{S}{hannon} observed in  \cite{shannon56} that feedback will not improve the capacity when communicating over a memory-less channel. This conclusion relies on the definition of capacity as a limiting case with arbitrary long blocks and no decoding delay constraints. Several authors have since then analysed different effects of feedback, see for instance\cite{horstein1963}, \cite{schalkwijk1966}, \cite{schalkwijk1966b}, \cite{gallager2010} and \cite{shayevitz2011}. The current paper is inspired by the interesting results of \cite{polyanskiy2011} where it is shown that
the Shannon-limit on $-1.6$dB energy per bit can be obtained even for the case of block length one, if a noise-free feedback channel is available. 
The obtained scheme however still has potentially unbounded decoding delay. 

To study the benefits of feedback in the case of finite block lengths the optimal strategies are presented in this article
in the case of a finite decoding delay constraint for a discrete time Gaussian channel as depicted in Fig. \ref{fig1} with a transmitted message $m\in \{0,1\}$. A dynamical programming scheme is described that finds the optimal strategies numerically by repeated one-dimensional minimizations. 
\tikzstyle{block} = [draw, fill=blue!20, rectangle, 
    minimum height=2.5em, minimum width=5em]
\tikzstyle{sum} = [draw, circle, node distance=1cm]
\tikzstyle{input} = [coordinate]
\tikzstyle{output} = [coordinate]
\tikzstyle{pinstyle} = [pin edge={to-,thin,black}]

   \begin{figure}[!ht]
   \begin{tikzpicture}[auto, node distance=2cm,>=latex']
    \node [input, name=input] {};
    \node [block, right of=input, node distance=2.7cm] (code) {Code};
    \node [sum, right of=code, node distance=2.0cm] (sum) {+};
    \node [block, right of=sum] (decode) {Decode};
    \draw [draw,->] (input) -- node {$m\in \{0,1\}$} (code);
    \node [output, right of=decode] (output) {};
    \node [input, above of=sum,node distance=0.8cm] (disturbance) {};
    \draw [->] (code) -- node {$x_k$} (sum);
    \draw ($(disturbance) + (0,2mm)$) node {$z_k$} ;
    \draw [->] (disturbance) -- (sum);
    \draw [->] (sum) -- node [name=y] {$y_k$} (decode);
    \draw [->] (decode) -- node [name=mhat] {$\widehat m$}(output);
    \draw [->] (y.south) -- ++(0,-1) node [below,left,yshift=-3mm] {feedback channel} -|  (code.south);
\end{tikzpicture}
  \caption{The system studied in the paper. The feedback channel is assumed  noise-free. The decoding has  a delay constraint $N$, i.e. $\widehat m$ should be produced after observing $y_1,\ldots,y_N$.}
   \label{fig1}
   \end{figure}
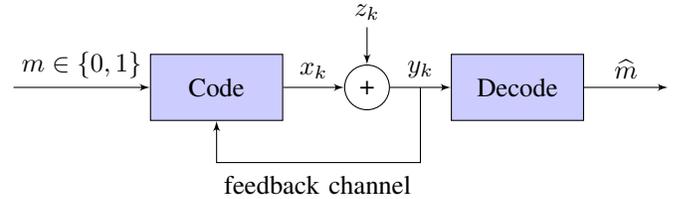

The computational method avoids the combinatorial explosion that a straight-forward approach of searching for strategies of the form $x(1, m),x(2, y(1),m),x(3, y(1),y(2),m),\ldots $, $x(N, y(1),\ldots,y(N-1),m)$ would lead to, where $y(k)$ is the channel output at time step $k$. The paper generalizes the result presented in \cite{bernhardsson2013} where the case $N=2$ was studied. 

We also study the problem of transmitting a message $m\in \{0,1\}$ over $M$ parallel analog Gaussian white noise feedback channels, with $m=0$ and $m=1$ being equally likely. The channels are given by 
 $$
y_i(k) =  x_i(k) + z_i(k), ~~~~~ i=1, ..., M, ~~~ k = 1, ..., N,
$$  
where $x(k)=(x_1(k), \dots, x_M(k))$ are the signal inputs, $y(k) = (y_1(k), ..., y_M(k))$ are the measurements at the receiver side, and $z(k) = (z_1(k), ..., z_M(k))\sim \gd(0,I)$ is Gaussian white noise.

The encoder/transmitter is restricted to transmit real numbers $x_i(k)$ over a finite time interval $k=1, ... N$, using side-information from a causal noise-free feedback channel, see Fig.~\ref{fig1}.
   
%

At time $k$ the encoder transmits the real vector
$x(k)=x(k, y^{k-1},m)$, where $y^{k-1} = \emptyset$ for $k=1$ and $y^{k-1} = \{y(1),\ldots,y(k-1)\}$ for $k>1$. We will analyze the case where up to $N$ transmissions are allowed, i.e.
\begin{align}
y(k) &= x(k, y^{k-1},m) + z(k), \quad k=1,\ldots,N.
\end{align}

Several authors have analyzed different effects of feedback, see for instance\cite{horstein1963}, \cite{schalkwijk1966}, \cite{schalkwijk1966b}, \cite{gallager2010}, \cite{shayevitz2011}, and \cite{polyanskiy2011}. This paper generalizes the result presented in \cite{bernhardsson2013} to the multi-dimensional case.

The contribution of this paper is the construction of the optimal encoder functions $x(k, y^{k-1},m)$ and a decoder scheme producing an estimate $\widehat m \in \{0,1\}$, that minimizes the bit error probability 
\[
P^e = \mathbf{Pr}(\widehat m \neq m)
\] 
and fulfils an average energy constraint 
\begin{align}
\underset{z_1, ..., z_N}{\E}\left(\sum_{k=1}^N |x(k)|^2\right)\leq S \label{S}
\end{align}
for a prescribed level $S>0$.

With no feedback, communicating one bit using the Gaussian vector channel of dimension $M$ at most $N$ times is
equivalent to using a scalar Gaussian channel $MN$ times. Let $\varphi(t)=(2\pi)^{-\frac{1}{2}}e^{-t^2/2}$ and $Q(a)=\int_a^{\infty}\varphi(t)dt$.
It is well known, see e.g. \cite{rappaport01}, that the optimal bit error rate without feedback is given by 
 \begin{align}
 P^e_{\textrm{no feedback}}&= Q\left(\sqrt{S}\right)\label{Petrad},  
 \end{align}
  which  can be achieved by antipodal signaling
 $x_1(1)=\pm \sqrt{S}$ and $x_l(k)=0$ for $(k,l)\ne (1,1)$. Without feedback there is no performance benefit with splitting the energy into several transmissions and no benefit in splitting the energy over several channels.

\section{Optimal Communication Strategies}
\subsection{Optimal Decoder}
Consider the decoder at the receiver side. 
Let $P(m=i \mid y^{k-1})$ be the (a posteriori) probability that the transmitted message is $m=i$ given the
measurments $y^{k-1}$, $p(y^{k-1} \mid m=i)$ the conditional probability density of $y^{k-1}$
given $m$, $P(m=i)$ the probability that $m=i$, and $p(y^{k-1})$ the probability density function of $y^{k-1}$.
Bayes' law gives the relation
\begin{equation}
\label{bayes}
P(m=i \mid y^{k-1}) =\frac{p(y^{k-1} \mid m=i)P(m=i)}{p(y^{k-1})}.
\end{equation}
Since $P(m=0) = P(m=1)=\frac{1}{2}$, we have that
\begin{equation}
\label{ratio}
\frac{P(m=1 \mid y^{k-1})}{P(m=0 \mid y^{k-1})} = \frac{p(y^{k-1} \mid m=1)}{p(y^{k-1} \mid m=0)}.
\end{equation}
Now define
\begin{equation}
\label{logratio}
l_k:= \log{\frac{p(y^{k-1} \mid m=1)}{p(y^{k-1} \mid m=0)}}.
\end{equation}
It is well known that the decoder that minimizes decoding error probability uses maximum likelihood detection and will therefore output the message $\widehat{m}=1$ if $l_k > 0$ and $\widehat{m}=0$ otherwise. 

Now we have that
\begin{align*}
p(&y^N  \mid m) =\\
				&= p(y(N), ..., y(2) \mid y(1), m)p(y(1) \mid m)\\
				&~~\vdots \\
				&= \prod_{k=1}^N p(y(k) \mid y^{k-1}, m)\\
				&= \prod_{k=1}^N p(z(k))
\end{align*}
Because of the Gaussian assumption of the noise $z$, we have 
\begin{align*}
&\log p(y^{N}\mid m) \\
&\quad = -\frac{N}{2}\log(2\pi)-\frac{1}{2}\sum_{k=1}^N |z(k)|^2\\
&\quad = -\frac{N}{2}\log(2\pi)-\frac{1}{2}\sum_{k=1}^N |y(k)-x(k, y^{k-1},m)|^2.
\end{align*}
To shorten notation introduce $u^1(k) := x(k, y^{k-1}, 1)$ and $u^0(k) := x(k, y^{k-1}, 0)$, where the  
the dependence of $y^{k-1}$ is suppressed.
The decoded bit $m$ is then determined by the sign of the log-likelihood ratio
\begin{align*}
l_{N+1}&= \log \frac{p(y^N \mid m=1)}{p(y^N \mid m=0)}\\
& = \frac{1}{2}\sum_{k=1}^N -\big|y(k)-u^1(k)\big|^2+ \big|y(k)-u^0(k)\big|^2 \\
\end{align*}
We note that the log-likelihood ratio $l_k$ satisfies the recursion
\begin{align*}
	l_{k+1} &= l_k -\frac{1}{2}\big|y(k)-u^1(k)\big|^2+ \frac{1}{2}\big|y(k)-u^0(k)\big|^2\\
		l_1 &= 0.
\end{align*}

Now combining (\ref{ratio}) and (\ref{logratio}), we get 
 \begin{align}
p_{k}^1&:=P(m=1\mid y^{k-1}) =  \dfrac{e^{l_k}}{e^{l_k}+1} = \dfrac{1}{e^{-l_k}+1}, \label{p1}\\
p_{k}^0&:=P(m=0\mid y^{k-1}) =  \dfrac{1}{e^{l_k}+1}. \label{p0}
\end{align}
This means that $l_k$ is a sufficient statistics for the receiver(decoder) to convey the information about $m$ given by
the measurements $y^{k-1}$. Note that if we know $y^{k-1}$, then we also know $z^{k-1}$ and vice versa, so
\begin{align*}
P(m=1\mid z^{k-1}) &= p_{k}^1,\\
P(m=0\mid z^{k-1}) &=p_{k}^0.
\end{align*}

\subsection{Optimal Encoder}
In this section, we will consider the optimal encoder in order to maximize the expected value of the probability that the decoder decodes the correct transmitted message. 

Note first that the law of iterated conditional expectations, 
\begin{equation}
\label{iter_cond}
 \underset{w}{\E}(x) = \underset{w,y}{\E}\left(\underset{x\mid y}{\E}(x\mid y)\right),
\end{equation}
implies that
	\begin{eqnarray}
		&\underset{m, z^N}{\E}& \left(|u^m(k)|^2\right) =\\
		&=&\underset{m, z^{k-1}}{\E}\left(|u^m(k)|^2\right) \label{reduced_z} \\
		&=& \underset{m, z^{k-1}}{\E}\left(\underset{|u^m(k)|^2\mid z^{k-1}}{\E}\left(|u^m(k)|^2\mid z^{k-1}\right)\right) \label{iter_cond1}\\
		&=& \underset{m, z^{k-1}}{\E}\left(\underset{|u^m(k)|^2\mid z^{k-1}}{\E}\left(|u^m(k)|^2\mid z^{k-1}, y^{k-1}\right)\right)\nonumber\\
		&& \label{z_spans_y}\\
		&=& \underset{m, z^{k-1}}{\E}\left(p_k^1 |u^1(k)|^2 + p_k^0 |u^0(k)|^2 \right) \label{expp}\\
		&=& \underset{z^{k-1}}{\E}\left(p_k^1 |u^1(k)|^2 + p_k^0 |u^0(k)|^2 \right) \label{mindepend}
	\end{eqnarray}
where (\ref{reduced_z}) follows from the fact that $z(k), ..., z(N)$ are independent of $m$ and $u^m(k)$, (\ref{iter_cond1}) follows from (\ref{iter_cond}), (\ref{z_spans_y}) follows from that fact that $y^{k-1}$ can be constructed from $z^{k-1}$, (\ref{expp}) follows from (\ref{p1}) and (\ref{p0}), and (\ref{mindepend}) follows from the fact that $p_k^1 |u^1(k)|^2 + p_k^0 |u^0(k)|^2$ and $z^{k-1}$ are independent of $m$.
Therefore, we have that
$$
S\geq \sum_{k=1}^N \E|u^m(k)|^2 = \sum_{k=1}^N  \E\left ( p_k^1 |u^1(k)|^2 +  p_k^0 |u^0(k)|^2\right)
$$

Now suppose that the transmitted message is $m=1$. Then we have that $y(k) = u^1(k) + z(k)$. The log-likelihood ratio $l_k$ when $m=1$ is the message to be transmitted is given by the recursion
\begin{equation*}
	\begin{aligned}
		l_{k+1} 	&= l_k -\frac{1}{2}\big|y(k)-u^1(k)\big|^2+ \frac{1}{2}\big|y(k)-u^0(k)\big|^2\\
					&= l_k -\frac{1}{2}\big|u^1(k) + z(k) -u^1(k)\big|^2\\
					&~~~ + \frac{1}{2}\big|u^1(k) + z(k)-u^0(k)\big|^2\\
					&= l_k + \frac{1}{2}\big|u^1(k) - u^0(k)\big|^2 + \big(u^1(k)-u^0(k)\big)^\intercal z(k) \\
	\end{aligned}
\end{equation*}
The message is correctly decoded if $l_{N+1}>0$. Similarly, the log-likelihood ratio $l_k$ when $m=0$ is the message 
to be transmitted is given by
\begin{equation*}
	\begin{aligned}
		l_{k+1} 	&= l_k -\frac{1}{2}\big|y(k)-u^1(k)\big|^2+ \frac{1}{2}\big|y(k)-u^0(k)\big|^2\\
					&= l_k -\frac{1}{2}\big|u^0(k) + z(k) -u^1(k)\big|^2\\
					&~~~+ \frac{1}{2}\big|u^0(k) + z(k)-u^0(k)\big|^2\\
					&= l_k - \frac{1}{2}\big|u^1(k) - u^0(k)\big|^2 + \big(u^1(k)-u^0(k)\big)^\intercal z(k) \\
	\end{aligned}
\end{equation*}
The message is correctly decoded if $l_{N+1}<0$.
The optimization criterion is hence to minimize the probability of error
\begin{align*}
\mathbf{Pr}(\textrm{error}) 	&=P(m=1)P(l_{N+1}<0 \mid m=1)\\
								&~~~ + P(m=0)P(l_{N+1}>0\mid m=0)\\
								&=\frac{1}{2}P(l_{N+1}<0\mid m=1)\\
								&~~~~ + \frac{1}{2}P(l_{N+1}>0\mid m=0).
\end{align*}
The optimization problem we want to solve is thus
\begin{equation}
\label{opt1}
\begin{aligned}
		\inf_{\{u^0(k), u^1(k)\}} 	~~& P(l_{N+1}<0\mid m=1) + P(l_{N+1}>0\mid m=0)\\
		\textup{s. t.}~~~	l_{k+1}&= l_k + \frac{(-1)^{m+1}}{2}\big|u^1(k) - u^0(k)\big|^2 \\
									&~~~ + \big(u^1(k)-u^0(k)\big)^\intercal z(k) \\
		l_1 &= 0\\
S &\ge \sum_{k=1}^N  \E\left ( p_k^1 |u^1(k)|^2 +  p_k^0 |u^0(k)|^2\right)
\end{aligned}
\end{equation}
\begin{theorem}
The optimization problem 
(\ref{opt1})
is equivalent to
\begin{equation}
\label{opt3}
\begin{aligned}
		\inf_{\{v(k)\}} 	~~& P(l_{N+1}<0\mid m=1) + P(l_{N+1}>0\mid m=0)\\
		\textup{s. t.}~~~	l_{k+1} &= l_k + \frac{(-1)^{m+1}}{2}\big|v(k)\big|^2 + v(k) z_1(k) \\
		l_1 &= 0\\
S &\ge \sum_{k=1}^N  \E\left ( p_k^0 p_k^1 |v(k)|^2 \right)
\end{aligned}
\end{equation}
where $v:\{1, .., N\}\rightarrow \mathbb{R}$ (the notation suppresses that $v(k)$ also depends on $l_k$). An optimal solution of (\ref{opt1}) can be obtained from an optimal solution of  (\ref{opt3}) by
setting $u_l^m(k) = 0$ for $l>0$, 
$u_1^1(k) = p_k^0 v(k)$, and $u_1^0(k) = -p_k^1 v(k)$.
\end{theorem}

\begin{proof}
Let
$$
u^1(k) - u^0(k) = u(k)
$$
for some fixed function $u(k)$. The minimum value of 
$$
p_k^1 |u^1(k)|^2 +  p_k^0 |u^0(k)|^2 = p_k^1 |u^1(k)|^2 +  p_k^0 |u^1(k) - u(k)|^2
$$
is obtained by taking the derivative with respect to $u^1(k)$ and we get
$$
p_k^1 |u^1(k)|^2 +  p_k^0 |u^1(k) - u(k)|^2 \ge p_k^0 p_k^1 |u(k)|^2
$$
where the minimum is attained for $u^1(k) = p_k^0 u(k)$ and $u^0(k) = u^1(k) - u(k) = -p_k^1 u(k)$.
Thus, optimization problem (\ref{opt1}) becomes
\begin{equation*}
\begin{aligned}
		\inf_{u_1, ..., u_N} 	~~& P(l_{N+1}<0\mid m=1) + P(l_{N+1}>0\mid m=0)\\
		\textup{s. t.}~~~	l_{k+1} &= l_k + \frac{(-1)^{m+1}}{2}\big|u(k)\big|^2 + u^\intercal(k) z(k) \\
		l_1 &= 0\\
S &\ge \sum_{k=1}^N  \E\left ( p_k^0 p_k^1 |u(k)|^2 \right)
\end{aligned}
\end{equation*}
Note that the term $u^\intercal(k) z(k)$ in the recursion of $l_k$ given $u(k)$ is a scalar Gaussian variable
with variance $|u(k)|^2 = |u_1(k)|^2 + \cdots + |u_M(k)|^2$. Thus, for any (optimal) choice of 
$u(k) = (u_1(k), ..., u_M(k))$, $u_i:\{1, .., N\}\rightarrow \mathbb{R}$ for $i=1, ..., M$, 
we can take $u^\star(k)= (u_1^\star(k), 0, 0, ..., 0)$ with 
$|u_1^\star(k)|^2 =  |u_1(k)|^2 + \cdots + |u_M(k)|^2$, which renders a recursion for $l_k$ with identical statistics as that of $u(k)$. By setting $v(k)=u_1^\star(k)$, we obtain the optimization problem (\ref{opt3}). Hence, (\ref{opt1}) and (\ref{opt3})
are equivalent by setting $u_l^m(k) = 0$ for $l>0$, $u_1^1(k) = p_k^0 v(k)$, and $u_1^0(k) = u_1^1(k) - v(k) = -p_k^1 v(k)$, for $k=1, ..., N$. This completes the proof.
\end{proof}

The result above shows that when the measurement noise $z_i(k)$ Êis Gaussian, independent, and identically distributed for $k = 1, 2, ..., N$, $i = 1, 2, ..., M$, it is optimal to spend all the energy on one channel. 

Note that since the cost function in (\ref{opt3}) is bounded and since the constraint in (\ref{opt3})
restricts $v(1), ..., v(N)$ to belong to a compact set, the infimum is attained. Note also that any optimal 
set of strategies $v(1), ..., v(N)$ will be such that 
$$S = \sum_{k=1}^N  \E\left ( p_k^0 p_k^1 |v(k)|^2 \right).$$
Hence, (\ref{opt3}) is equivalent to
\begin{equation}
\label{opt4}
\begin{aligned}
		\min_{v(1), ..., v(N)} 	~~& P(l_{N+1}<0\mid m=1) + P(l_{N+1}>0\mid m=0)\\
		\textup{s. t.}~~~	l_{k+1} &= l_k + \frac{(-1)^{m+1}}{2}\big|v(k)\big|^2 + v(k) z(k) \\
		l_1 &= 0\\
S &= \sum_{k=1}^N  \E\left ( p_k^0 p_k^1 |v(k)|^2 \right)
\end{aligned}
\end{equation}
 
\begin{theorem}
There exists a nonnegative real number
$\lambda$ such that (\ref{opt4}) is equivalent to
\begin{equation}
\label{opt}
\begin{aligned}
\min_{v(1), ..., v(N)} 	~~~ P&(l_{N+1}<0\mid m=1) + P(l_{N+1}>0\mid m=0)\\ 
	&+\lambda\left(\sum_{k=1}^N  \E\left ( p_k^0 p_k^1 |v(k)|^2 \right)  - S\right) \\
		\textup{s. t.}~~~	l_{k+1} &= l_k + \frac{(-1)^{m+1}}{2}\big|v(k)\big|^2 + v(k) z(k) \\
		l_1 &= 0
\end{aligned}
\end{equation}
\end{theorem} 
\begin{proof}
Let $v(k) = v_k(\lambda)$, $k=1, ..., N$, be optimization variables depending on $\lambda$.  
Any optimal set of variables $v_1^\star(\lambda), ..., v_N^\star(\lambda)$ to (\ref{opt}) will be such that $\sum_{k=1}^N \E\left ( p_k^0 p_k^1 (v_k^\star(\lambda))^2 \right)$ goes from $\infty$ to 0 as $\lambda$ goes from 0 to $\infty$. It's not hard to verify that
$\E\left ( p_1^0 p_1^1 (v_1^\star(\lambda))^2 \right), ...,$ $\E\left ( p_N^0 p_N^1 (v_N^\star(\lambda))^2 \right)$
are continuous in $\lambda$ because of the expectation operator(which is a smoothing integral). Thus, there must exist a nonnegative real number $\lambda = \lambda_0$ such that 
\begin{equation}
\label{eqconstraint}
\sum_{k=1}^N \E\left ( p_k^0 p_k^1 (v_k^\star(\lambda_0))^2 \right)=S
\end{equation}

Since  $v_1^\star(\lambda_0), ..., v_N^\star(\lambda_0)$ minimize the objective function in (\ref{opt}) and at the same
time fulfills the equality constraint (\ref{eqconstraint}), it is also the optimal solution to (\ref{opt4}). This completes the proof.
\end{proof}

Note that since $v(k)$ is a function of $y^{k-1}$, they only have access to $(l_{1}, ..., l_{k})$ and no access to $(l_{k+1}, ..., l_{N+1})$.
Thus, the optimization problem (\ref{opt}) can be posed as a stochastic dynamic programming problem according to
$$
\min_{v(1), ..., v(N)} \underset{z(1), ..., z(N)}{\E}\left( g_{N+1}(l_{N+1}) + \sum_{k=1}^N g_k(l_k, v(k), z(k))\right)
$$
with 
$$g_{N+1}(l_{N+1}) = P(l_{N+1}<0\mid m=1) + P(l_{N+1}>0\mid m=0)$$ 
and $g_k(l_k, v(k), z(k)) = \lambda p_k^0 p_k^1 (v(k))^2 $,
subject to $l_1 =0$ and the dynamics
\begin{align*}
l_{k+1}
&= f^m_k(l_k, v_k, z_k) \\
&=l_k + \frac{(-1)^{m+1}}{2}\big| v(k)\big|^2 + v(k) z(k)\\
\end{align*}
which has a solution of the form $v(k)=\mu_k(l_k)$ that only depends on the current state $l_k$(\cite{bertsekas:dp1}). The problem can be solved according to the dynamics programming recursion
\begin{equation*}
	\begin{aligned}
		J_k(l_k) 	&= \min_{\mu_k}~ \underset{z(k)}{\E} \Big(g_k\big(l_k, \mu_k(l_k), z(k)\big) \\
							& ~~~~~~~~~~~~~ + J_{k+1}\big(f^m_k(l_k, \mu_k(l_k), z(k))\big) \Big) 
	\end{aligned}
\end{equation*}

It is also easy to convince oneself that the first transmission should be antipodal, i.e. $\E(x(1))=\frac{1}{2}(x(1,0)+x(1,1))=0$. This can be seen from the fact that a nonzero constant $\E(x(1))$ does not carry any information and just wastes energy since $\E(|x(1)|^2) = \E(|x(1)-\E(x(1))|^2) + (\E(x(1)))^2$.

\begin{figure}[!ht]
\includegraphics[width=1\hsize]{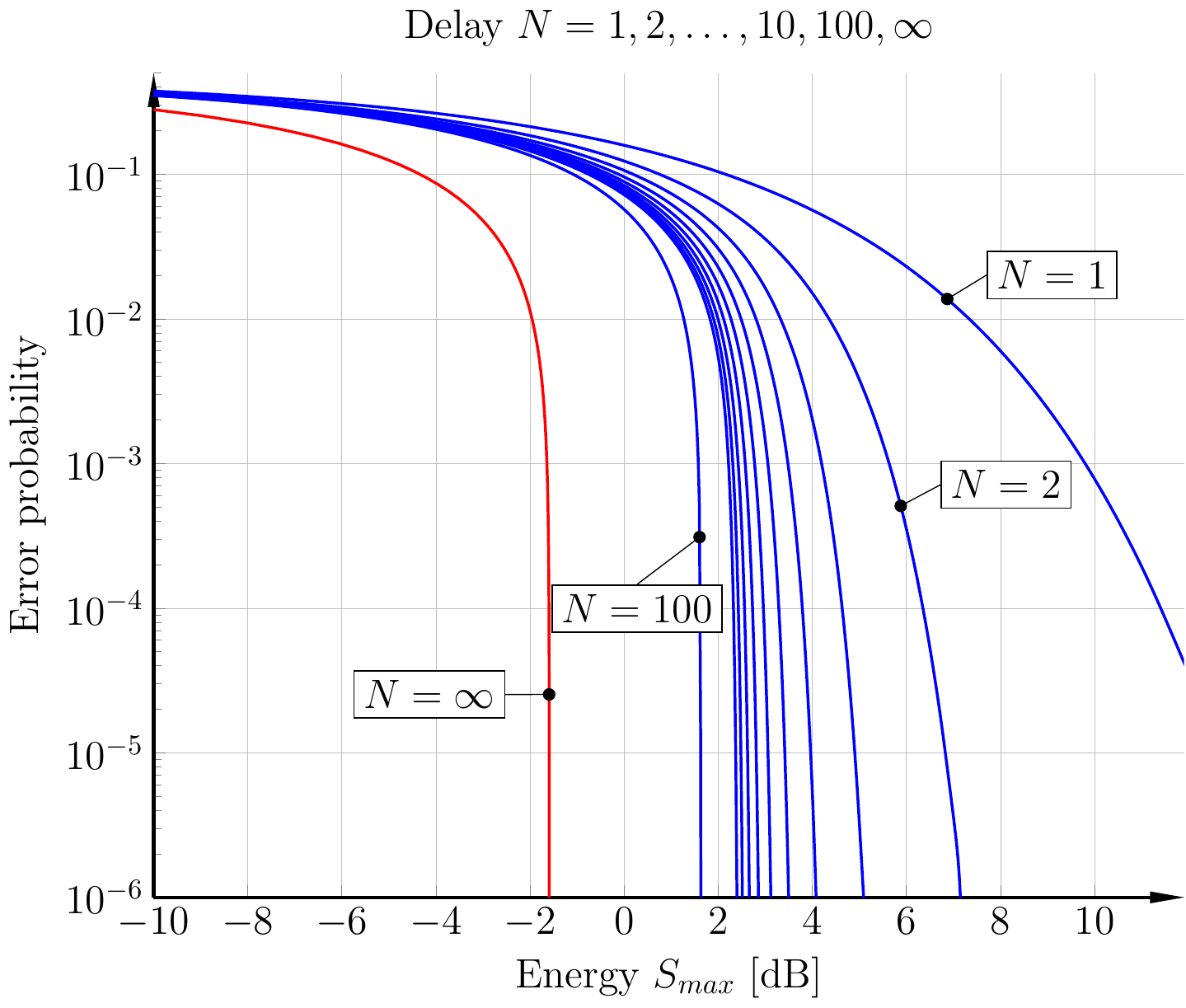}
\caption{Minimum bit error rate versus expected energy for delay constraint $N=1,2,\ldots,10, 100$, together with the Shannon bound without delay constraint.}
\label{fig2}
\end{figure}

\section{Results}

\begin{figure}[!ht]
\includegraphics[width=1.1\hsize]{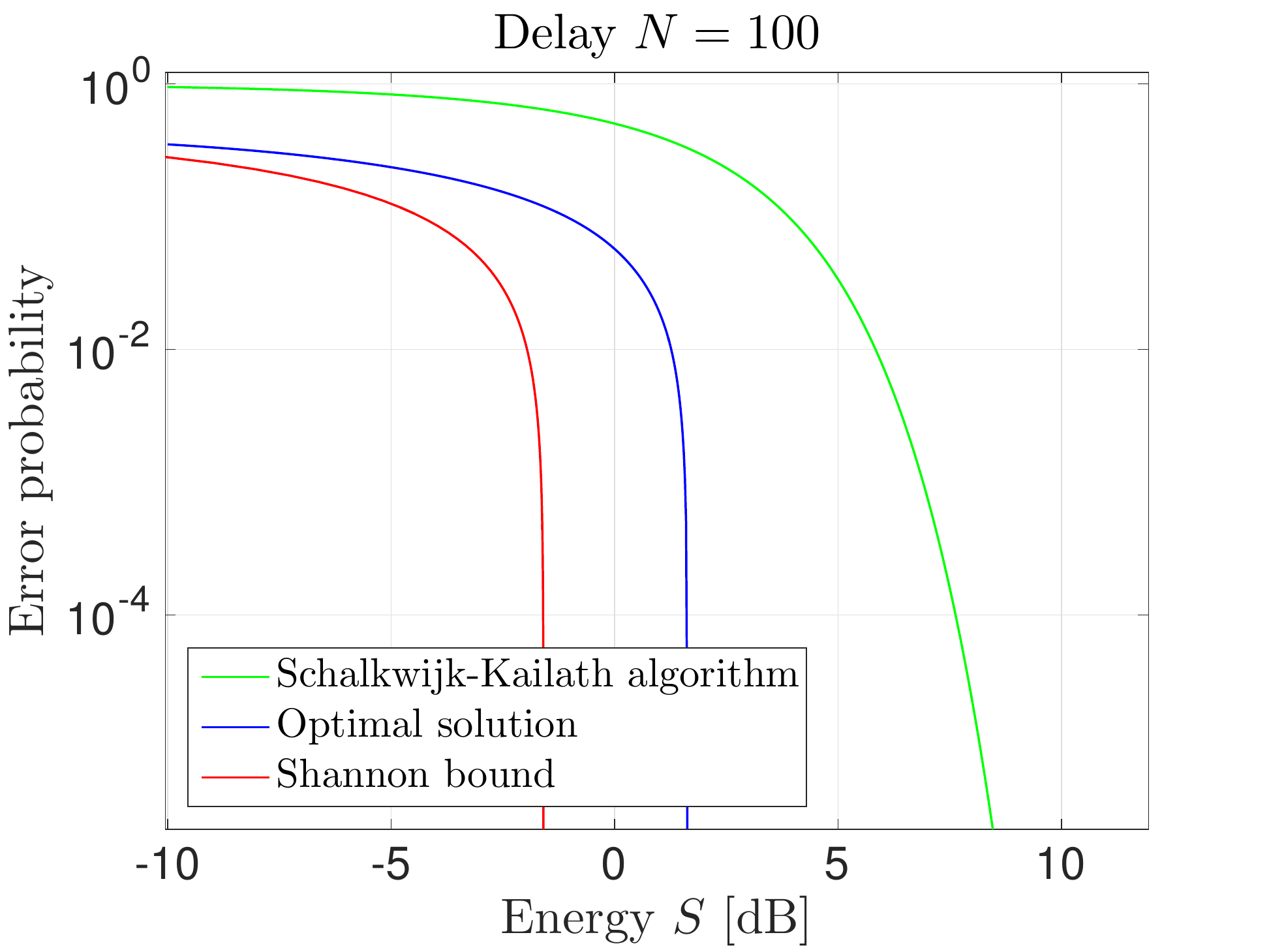}
\caption{Minimum bit error rates versus expected energy for delay constraint $N=100$, together with the Shannon bound without delay constraint.}
\label{S-K}
\end{figure}

 \begin{figure}[!t]
\centering
\includegraphics[width=1.1\hsize]{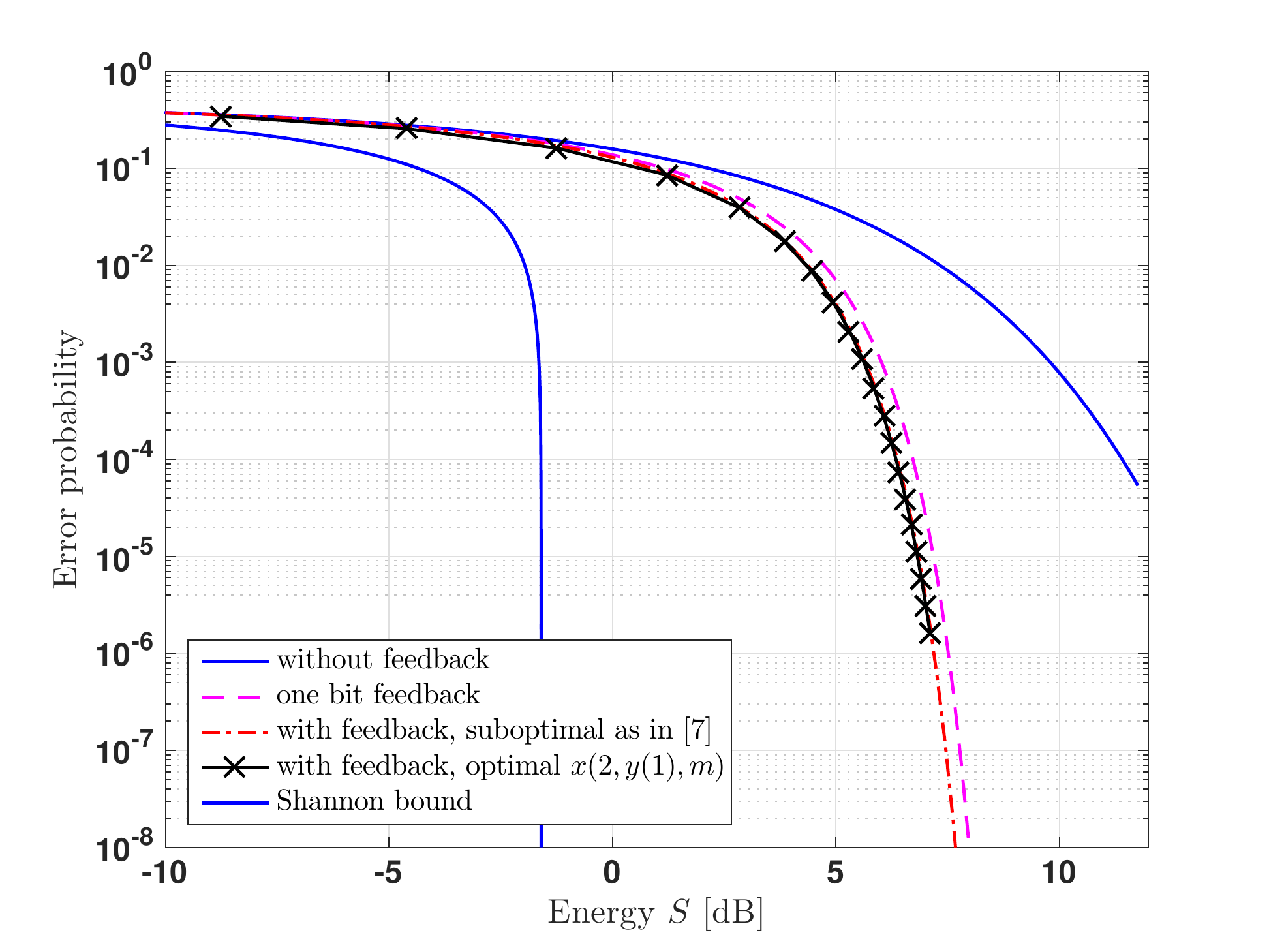}
\caption{Bit error probability versus average power:  Optimal transmission without use of feedback (full), one-bit feedback scheme (dashed) suboptimal feedback scheme (dash-dotted), optimal feedback scheme (full-x), Shannon bound for infinite-block transmissions (full). Notice the significant performance gain with feedback, even using only one-bit feedback. }
\label{bervspower}
\end{figure}

In this section we will show some results for the one-dimensional case ($M=1$) and make some comparisons with other transmission strategies.

Fig.~\ref{fig2} shows the error probability as a function of the energy budget $S = S_{max}$ for the optimium schemes with delay constraint $N=1,2,\ldots,10, 100$. Note that much of the performance difference between the scheme without feedback ($N=1$) and the Shannon bound for infinite block lengths (red curves) are recovered for $N=10$. The computational effort grows roughly linearly with $N$, so even longer delay constraints can easily be computed. This can be compared with the Schalkwijk-Kailath algorithm in \cite{schalkwijk1966} for the case $N=100$ as an example. As we can see in Fig. \ref{S-K}, the Schalkwijk-Kailath algorithm is not optimal and the optimal solution found in this paper is superior.

To get a glimpse of how the optimal strategies look like, we have studied the case $N=2$. Figures \ref{bervspower} -- \ref{bervspowerzoom} compare achievable performance for optimal transmission without use of feedback (top blue) and optimal transmission with use of feedback(black). Also shown is a suboptimal feedback scheme (red) corresponding to that used in  \cite{polyanskiy2011}. There is a significant performance gain of many dBs using feedback. The performance gain increases with SNR.
The suboptimal scheme from \cite{polyanskiy2011} is rather close to optimal, except for the low SNR regime where the optimal scheme outperforms the suboptimal with some tenths of dBs.
 Notice also that the feedback scheme obtainable with one-bit feedback (red dashed) captures most of the performance gain with feedback.
 The one-bit feedback scheme was obtained by assuming the feedback to give information about whether or not $|y(1)|\leq a$.  The level $a$ was found by straight-forward search. We have not been able to prove that this is the optimal use of the one-bit feedback channel.

\begin{figure}[!t]
\centering
\includegraphics[width=1.1\hsize]{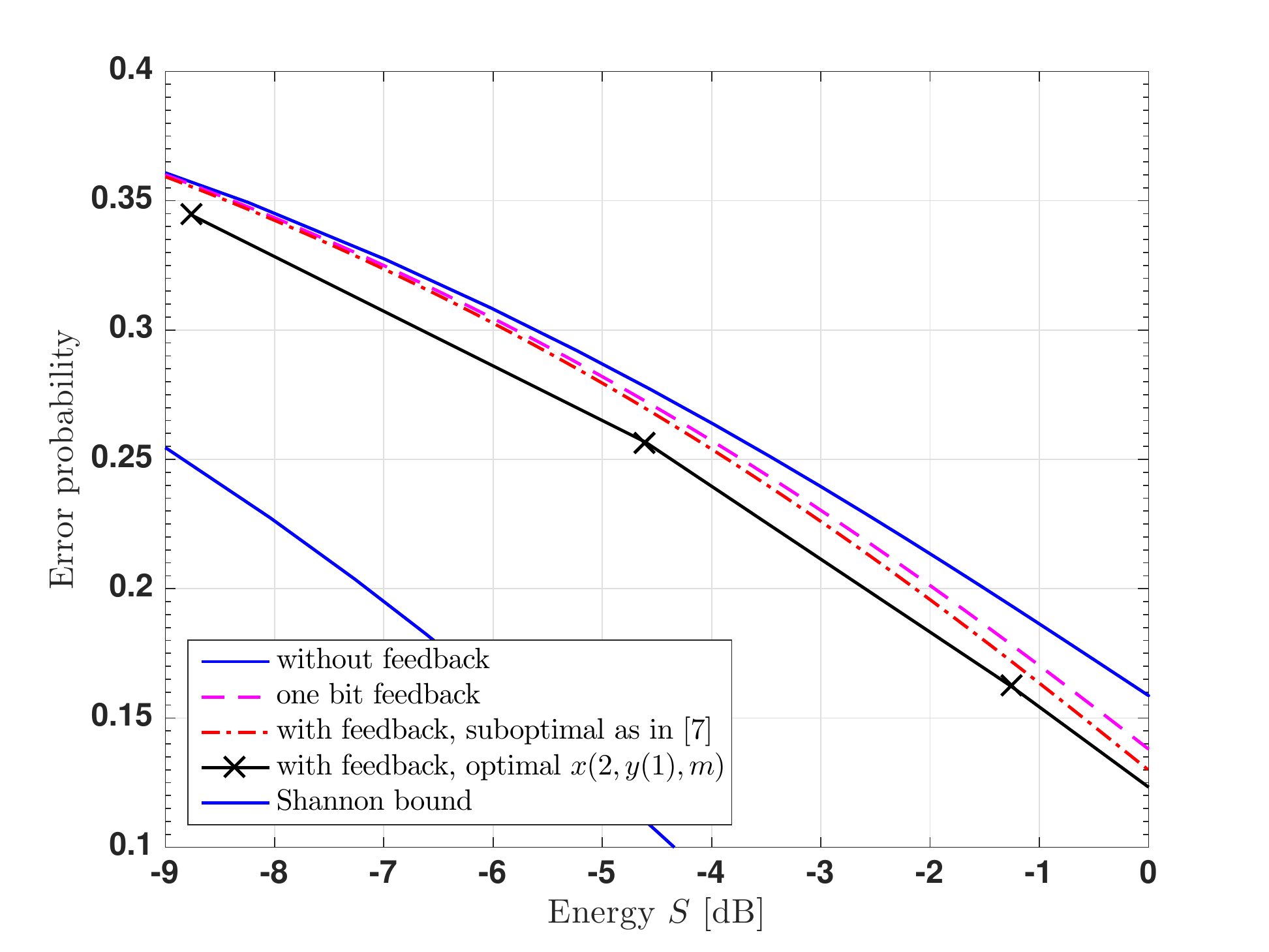}
\caption{Zoom of previous figure. There is a performance cost of $0.5-1$dB with the one-bit feedback channel, compared to using an infinite-capacity feedback channel.}
\label{bervspowerzoom}
\end{figure}

   \begin{figure}[!t]
   \centering
   \includegraphics[width=1.1\hsize]{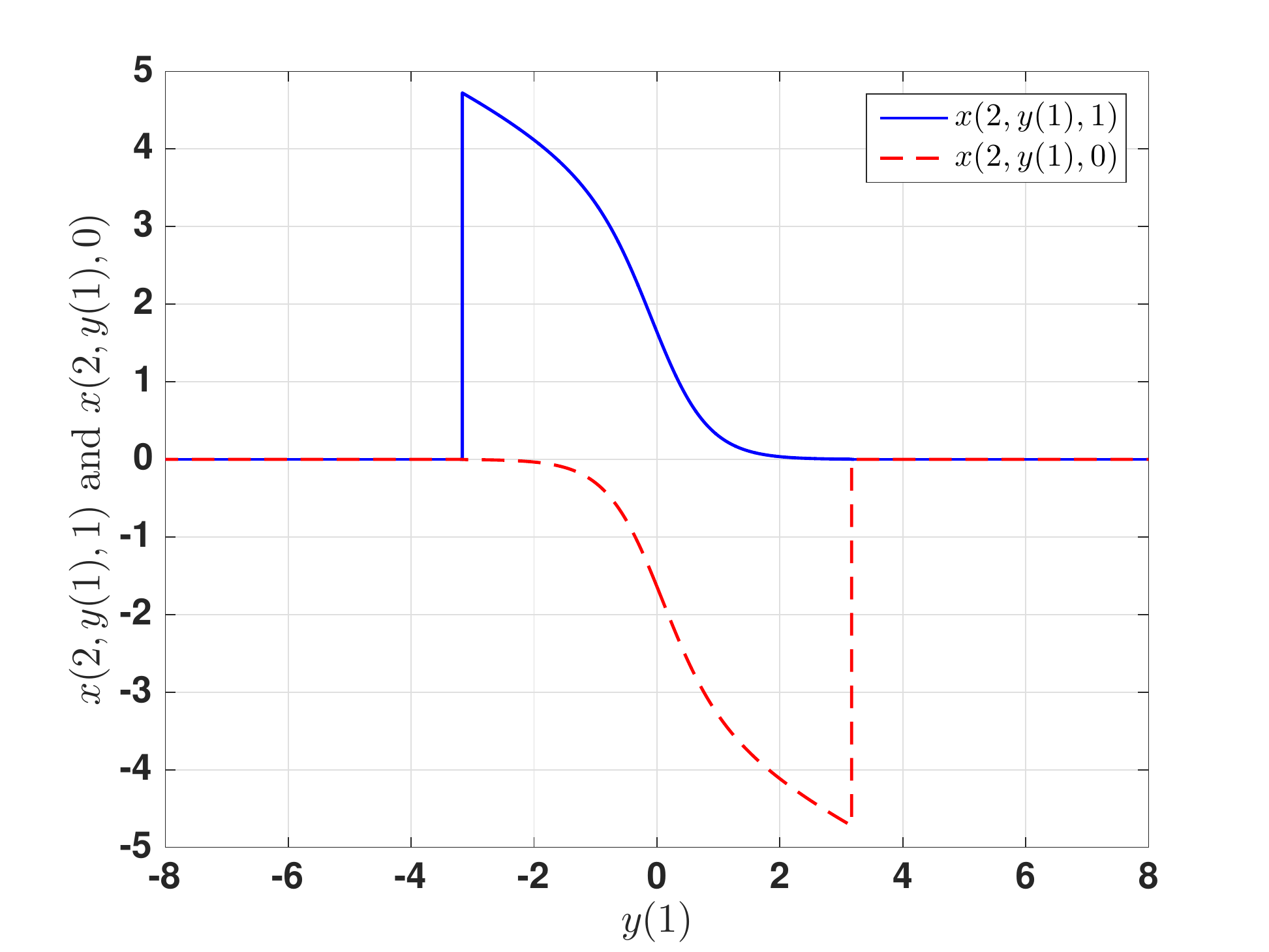}
   \caption{The optimal functions $x(2, y(1), 1)$ and $x(2, y(1), 0)$ when $S=S_{max}=2.42$, $x(1, 1)=1.19$, and $x(1,0) = -1.19$.}
   \label{u1y}
   \end{figure}

The optimal use of power in the second transmission, determined by $x(2, y(1), m)$, is interesting. The 
function $x(2, y(1), m)$ turns out to be discontinuous, showing that the second transmission should not be used
if the first output $y(1)$ is far away from zero. The discontinuity is manifested mostly in the low SNR regime, for high SNR the discontinuity threshold moves to very high levels of $y(1)$, corresponding to turning off the 2nd transmission only at extremely unlikely outcomes from the first transmission. Note that for low SNR the second transmission is used mainly when $y(1)$ is close to zero.
A majority of the power is used for the first transmission.
The optimal $x(2, y(1), 1)$ and $x(2, y(1), 0)$ for $S=S_{max}=2.42$ is illustrated in Fig.~\ref{u1y}.

 

\section{Conclusion}
We have studied the problem of communicating one bit over a MIMO analog white Gaussian noise channel with noiseless feedback from the encoder to the decoder. A delay constraint is imposed by allowing a maximum number of channel usage.
Under the assumption that the measurement noise variances are equal along the sub channels, we have shown that the problem can be reduced to transmitting one bit over a scalar feedback channel. In particular, since it has been previously shown \cite{bernhardsson2013}  that using the scalar channel twice with feedback is superior to using the channel twice without feedback, we conclude that communicating one bit over a MIMO channel with feedback is superior to that without feedback.
Future research could consider the problem of communicating one bit over more general MIMO feedback channels with interference between the sub channels. Also, the challenging problem of transmitting a multiple number of bits over a  channel with feedback under delay constraint is still open.

%
%
%
%

%
%
%
%
\bibliographystyle{IEEEtran}
\bibliography{IEEEabrv,./references}

\begin{thebibliography}{10}
\providecommand{\url}[1]{#1}
\csname url@samestyle\endcsname
\providecommand{\newblock}{\relax}
\providecommand{\bibinfo}[2]{#2}
\providecommand{\BIBentrySTDinterwordspacing}{\spaceskip=0pt\relax}
\providecommand{\BIBentryALTinterwordstretchfactor}{4}
\providecommand{\BIBentryALTinterwordspacing}{\spaceskip=\fontdimen2\font plus
\BIBentryALTinterwordstretchfactor\fontdimen3\font minus
  \fontdimen4\font\relax}
\providecommand{\BIBforeignlanguage}[2]{{%
\expandafter\ifx\csname l@#1\endcsname\relax
\typeout{** WARNING: IEEEtran.bst: No hyphenation pattern has been}%
\typeout{** loaded for the language `#1'. Using the pattern for}%
\typeout{** the default language instead.}%
\else
\language=\csname l@#1\endcsname
\fi
#2}}
\providecommand{\BIBdecl}{\relax}
\BIBdecl

\bibitem{shannon56}
C.~Shannon, ``The zero error capacity of a noisy channel,,'' \emph{Transactions
  on Information Theory}, vol.~2, pp. 8--19, sep 1956.

\bibitem{horstein1963}
M.~Horstein, ``Sequential transmission using noiseless feedback,''
  \emph{Information Theory, IEEE Transactions on}, vol.~9, no.~3, pp. 136 --
  143, jul 1963.

\bibitem{schalkwijk1966}
J.~Schalkwijk and T.~Kailath, ``A coding scheme for additive noise channels
  with feedback--i: No bandwidth constraint,'' \emph{Information Theory, IEEE
  Transactions on}, vol.~12, no.~2, pp. 172 -- 182, apr 1966.

\bibitem{schalkwijk1966b}
J.~Schalkwijk, ``A coding scheme for additive noise channels with feedback--ii:
  Band-limited signals,'' \emph{Information Theory, IEEE Transactions on},
  vol.~12, no.~2, pp. 183 -- 189, apr 1966.

\bibitem{gallager2010}
R.~Gallager and B.~Nakiboglu, ``Variations on a theme by {S}chalkwijk and
  {K}ailath,'' \emph{Information Theory, IEEE Transactions on}, vol.~56, no.~1,
  pp. 6 --17, jan. 2010.

\bibitem{shayevitz2011}
O.~Shayevitz and M.~Feder, ``Optimal feedback communication via posterior
  matching,'' \emph{Information Theory, IEEE Transactions on}, vol.~57, no.~3,
  pp. 1186 --1222, march 2011.

\bibitem{polyanskiy2011}
Y.~Polyanskiy, H.~Poor, and S.~Verdu, ``Minimum energy to send k bits through
  the {G}aussian channel with and without feedback,'' \emph{Information Theory,
  IEEE Transactions on}, vol.~57, no.~8, pp. 4880 --4902, aug. 2011.

\bibitem{bernhardsson2013}
B.~Bernhardsson, ``Using a {G}aussian channel twice,'' in \emph{Proc.IEEE
  International Symposium on Information Theory (ISIT, 2013)}, 2013.

\bibitem{rappaport01}
T.~Rappaport, \emph{Wireless Communications: Principles and Practice},
  2nd~ed.\hskip 1em plus 0.5em minus 0.4em\relax Upper Saddle River, NJ, USA:
  Prentice Hall PTR, 2001.

\bibitem{bertsekas:dp1}
D.~P. Bertsekas, \emph{Dynamic Programming and Optimal Control}, 2nd~ed.\hskip
  1em plus 0.5em minus 0.4em\relax Athena Scientific, 2000.

\end{thebibliography}

%

\end{document}